\documentclass{article}
\title{Tight Lower Bounds for Unequal Division}
\author{Andrew Lohr}
\usepackage{ amssymb }
\usepackage{synttree}

\usepackage{amsthm}

\newcommand{\caseTitle}[1]{{\bf #1}\\}
\newtheorem{buildup}{Lemma}
\newtheorem{construction}{Theorem}
\newtheorem{define}{Definition}
\newtheorem{corr}{Corollary}
\begin{document}
\maketitle
\section{Introduction}
Alice and Bob want to cut a cake; however, in contrast to the usual problems of fair division, they want to cut it unfairly. More precisely, they want to cut it in ratio $(a:b)$. (We can assume gcd(a,b)=1.) Let f(a,b) be the  number of cuts this will take (assuming both act in their own self interest).
It is known that $f(a,b) \le \lceil \lg(a+b)\rceil$. We show that (1) for all a,b, $f(a,b) \ge lg(lg(a+b))$ and (2) for an infinite number of (a,b), $f(a,b) \le 1+lg(lg(a+b)$.

The problem of discrete, unequal division already has a solution given in \cite{robertsonwebb}, which acheives a ratio $(a,b)$ in at most $\lceil \lg(a+b) \rceil$ steps. This procedure works to effectively halve the sum of ratio with each cut by having a peice cut off that is near half of the current sum of the ratio. Then the other person takes the piece if they value it more than the person who cut it. So that no matter how the  piece is allocated, that ratio falls by a factor of approximately two. 
The following tree shows the possible sequences of ratios to be considered when applying this procedure to the ratio $(9,8)$. Note, it uses 5 cuts in worst case.
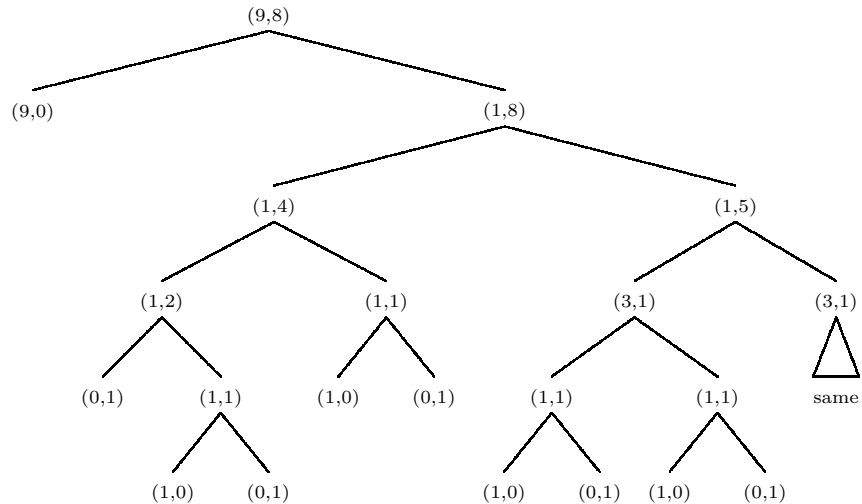
\begin{scriptsize}
\begin{center}
{\synttree [(9,8) 
[(9,0)]
[(1,8)
[(1,4)[(1,2)[(0,1)][(1,1)[(1,0)][(0,1)]]][(1,1)[(1,0)][(0,1)]]]
[(1,5)[(3,1)[(1,1)[(1,0)][(0,1)]][(1,1)[(1,0)][(0,1)]]][(3,1)[.x same]]]]]}
\end{center}
\end{scriptsize}

\section{Definitions}
We call a ratio $(a,b)$ in lowest terms if $\gcd(a,b) = 1$.  Notice that for the purposes of allocation, $(a,b)$ is equivalent to $(b,a)$ we'll use them interchangably.
\begin{define}
A standard, discrete protocol is a protocol in which each person involved is able to either respond with what value they place on a particular piece, or cut off a piece of a given value.
\end{define}
\begin{define}
An $(a,b)$-division is a standard protocol involving Alice and Bob such that Alice receives at least $\frac{a}{a+b}$ and Bob receives at least $\frac{b}{a+b}$.
\end{define}
We will show later that a more restricted form of standard protocols is all that need be considered for $(a,b)$-divisions.
\begin{define}
Let $f(a,b)$ be the smallest number of worst-case cuts needed for an (a,b)-division.
\end{define}
\begin{define}
$ProcA(c)$ is the procedure in which one of the two people (selected arbitrarily) cuts off a piece that they value at $c$. Then, if the other person values that piece $>c$, they take it, otherwise the person who cut it off takes it.
\end{define}
\begin{define}
$ProcB(c)$ is the procedure in which the person who is owed less of the whole, cuts off a piece that they value at $c$. Then, the other person takes that peice if they value it more than $c$, otherwise they take the other piece, which they neccesarily value at $\ge 1-c$
\end{define}

\begin{define}
The binary operations on ratios $*_1$, $*_2$, and $*_3$ are:
\end{define}
\begin{center}
\begin{math}
\begin{array}{lcl} 
(a_1,b_1) *_1 (a_2,b_2) &=&((a_1 +b_1)  a_2, (a_2 + b_2)  b_1)\\
(a_1,b_1) *_2 (a_2,b_2) &=&(a_1  a_2,  a_1  a_2 + b_2  a_1 + a_2  b_1)\\
(a_1,b_1) *_3 (a_2,b_2) &=&(a_1  b_2 + b_1  a_2 + b_2  b_1, b_2  b_1)
\end{array}
\end{math}
\end{center}
Note, that if either ratio given as an argument is scaled, then that merely causes the result to be scaled by the same factor, so, these operators are independent of representatives of the ratios used.

\section{Examination of $(a,b)$-Divisions}
\subsection{A Motivating Example}
We are able to see, however, that this logarithmic bound is not always tight. At certain times it is possible to acheive a ratio in many fewer cuts than the previous method by selecting the cutoff value so that both possibilities for the remaining ratio to be divided is not in lowest term. So, when reduced to lowest terms, they are much smaller.

A simple example of this approach doing better than the cut-near-halves algorithm is given for the ratio $(9,8)$:
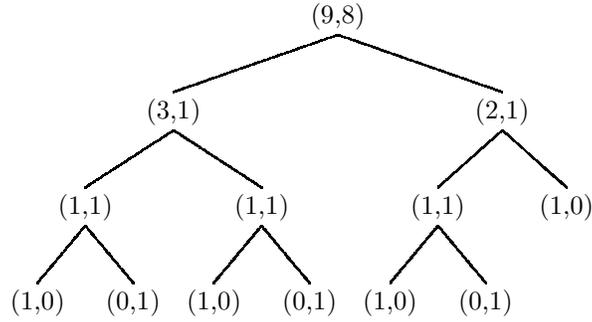
\begin{center}
{\synttree [(9,8) [(3,1)[(1,1)[(1,0)][(0,1)]][(1,1)[(1,0)][(0,1)]]][(2,1)[(1,1)[(1,0)][(0,1)]][(1,0)]]]}
\end{center}
We see that by having one person cut off $\frac{5}{17}$ initially. Depending on whether the other person thinks it is less than or greater than $\frac{5}{17}$, we get the subproblems $(9,3)$ and $(4,8)$ which are equivalent to $(3,1)$ and $(2,1)$ respectively. 

So, by selecing the cutoff carefully, we were able to acheive $(9,8)$ in only 3 cuts instead of 5 cuts. By computer search we also found many much larger ratios, for example, with six cuts, we can get a $(58470565,72019008)$-division  instead of the twenty eight cuts that cut near halves would require. We will be investigating a generalization of this type of selection of the amount to be cut off, and also give a lower bound on $f(a,b)$ in terms of $a+b$ (of course, assuming $(a,b)$ is in lowest terms).

\subsection{Foundation}
Although standard protocols allow for arbitrary sequences of evaluation of pieces and asking a person to cut off a piece of a given size, our situation can only benefit from a more restricted set of operations. Since we only have two people, the only evaluations it will be helpful to ask for are those from the person who didn't make the cut. Also, after a piece is cut, and evaluations are made, one of the pieces must be allocated to a person, and the procedure continues only on the other piece.

\begin{buildup}
Any optimal $(a,b)$-division can be rewritten as sequences of $ProcA(c)$ and $ProcB(c)$, each time picking some $c$
\end{buildup}
Suppose that the ratio to be divided is $(a,b)$ wlog $a<b$ (if $a=b$, do cut-and-choose which is $ProcA(\frac{1}{2})$). So, the basic operation we are left with takes three forms, depending on what fraction $\frac{k}{d}$ we ask a person to make first. \\
\caseTitle{Case 1: $\frac{a}{a+b}<\frac{k}{d} <\frac{b}{a+b}$} Notice in this case that both pieces produced are $> \frac{a}{a+b}$ so, the piece to be allocated can't be to Alice. If Bob evaluates the piece cut off to be  $\ge \frac{k}{d}$ he takes that piece, otherwise he takes the other piece, which is necessarily $\ge \frac{d-k}{d}$. Then, they procede to divide the unallocated piece in the ratio either $(a*d,b*d-k*(a+b))$ or $(a*d,k*(a+b)-a*d)$ depending wether Bob took the piece Alice valued at $\frac{k}{d}$ or the piece she valed at $\frac{d-k}{d}$, respectively\\
\caseTitle{Case 2: $\frac{k}{d} \le \frac{a}{a+b}<\frac{b}{a+b}$} Notice that the piece that is not cut off is greater than either person's due share, so, the only possibility is that the piece allocated is the one that Alice cut off. If Bob evaluates that piece to be $\ge \frac{k}{d}$ then he is allocated that piece, otherwise Alice is allocated that piece. Then, The ratio in which to divide the unallocated is either $(a*d,b*d-k*(a+b))$ or $(a*d-k*(a+b),b*d)$ depending on whether Bob or Alice got the piece, respectively. \\
\caseTitle{Case 3: $\frac{a}{a+b}<\frac{b}{a+b} <\frac{k}{d} $}
This is symmetric to Case 2 because we have $1-\frac{k}{d}\le \frac{a}{a+b}<\frac{b}{a+b}$ so the only difference from Case 2 is that we consider the remaining piece that was not cut off, instead of the piece that was.

\begin{corr} \label{goingdown}
We get one of the following two relations, depending on our choice of $\frac{k}{d}$ at each operation
\begin{center}
\begin{math}
f(a,b) = 1 + \max\{f(a*d,b*d-k*(a+b)),f(a*d,k*(a+b)-a*d)\}\, \,  \frac{a}{a+b}<\frac{k}{d}<\frac{b}{a+b} \linebreak
f(a,b) = 1+ \max\{f(a*d-k*(a+b),b*d),f(a*d,b*d-k*(a+b))\}\, \,\, \, \frac{k*(a+b)}{d}\le a,b
\end{math}
\end{center} 
\end{corr}
\begin{proof}
Each operation takes a single cut, and leaves you with one of two ratios left to divide, depending on the preferences of the non-cutter. Since there is no control of their preference in the protocol, either ratio could need to be acheived by the protocol.
\end{proof}
\begin{buildup} \label{buildupper}
Let $(a_1,b_1)$ and $(a_2,b_2)$ be two ratios. Take $(a,b)$ a ratio in lowest terms that can reduce to either $(a_1,b_1)$ and $(a_2,b_2)$ as described in Corollary ~\ref{goingdown}. Then $(a,b)$ is one of $(a_1,b_1) *_1 (a_2,b_2)\,$,$\,(a_1,b_1) *_2 (a_2,b_2)\,$, or $(a_1,b_1) *_3 (a_2,b_2)$
\end{buildup}
\begin{proof}
wlog,  take $b_1 a_2 \ge a_1 b_2$.  
Let $\frac{k}{d}$ the piece that is cut off. Then, since neither $\frac{b_2}{a_2}= \frac{b}{a}$ nor $\frac{b}{a} = \frac{b_1}{a_1}$ is not possible by Corollary ~\ref{goingdown}, there are three cases for ordering the ratios: \\
\caseTitle {Case 1 : $\frac{b_2}{a_2}< \frac{b}{a} < \frac{b_1}{a_1}$}
 In this case, we have to be using the second case of Corollary ~\ref{goingdown} in which $k<a$,$k<b$. This means that there are some factors $s_1,s_2$ such that $s_1 * (a_1,b_1) = (a*d-k*(a+b),b*d)$ and $s_2 * (a_2,b_2) = (a*d,b*d-k*(a+b))$ so, we get the system:
\begin{center}
\begin{math}
\begin{array}{lcl} 
s_1 a_1 & = & a*d-k*(a+b) \\
s_1  b_1 & = & b*d \\
s_2  a_2 & = & a*d \\
s_2  b_2 & = &  b*d-k*(a+b)
\end{array}
\end{math}
\end{center}
which gets us that
\begin{center}
\begin{math}
\begin{array}{lcl} 
a &= &(a_1 +b_1)  a_2\\
b &= &(a_2 + b_2)  b_1\\
\end{array}
\end{math}
\end{center}
so, in this case, the only ratio that can depend on $(a_1,b_1)$ and $(a_2,b_2)$ is $( (a_1 +b_1) * a_2, (a_2 + b_2) * b_1) = (a_1,b_1) *_1 (a_2,b_2)$\\
\caseTitle{Case 2: $\frac{b}{a}  < \frac{b_2}{a_2} \le \frac{b_1}{a_1}$}
In this case, we know as well that we must have a cutoff that makes b decrease in both cases, as there is no hope of ending up in either ratio if we decrease a but not b, this means that we are int the case of Corollary ~\ref{goingdown} with $\frac{k}{d} > \frac{a}{a+b}$. Note, we have two choices of $\frac{k}{d}$ symmetric about $\frac{1}{2}$ but, a and b are still unique up to a common scaling factor.
\begin{center}
\begin{math}
\begin{array}{lcl} 
s_1 a_1 &=& a*d\\
s_1  b_1 &=& k*(a+b)-a*d \\
s_2  a_2 &=& a*d \\
s_2  b_2 &=& b*d-k*(a+b)
\end{array}
\end{math}
\end{center}
which gets us a solution that
\begin{center}
\begin{math}
\begin{array}{lcl} 
a &=& a_1 a_2\\
b &=& a_1 a_2 + b_2 a_1 + a_2 b_1\\
\end{array}
\end{math}
\end{center}
So, the only ratio that can depend on $(a_1,b_1)$ and $(a_2,b_2)$ is $(a_1 * a_2,  a_1 * a_2 + b_2 * a_1 + a_2 * b_1)=(a_1,b_1) *_2 (a_2,b_2)$\\
\caseTitle {Case 3:$ \frac{b_2}{a_2} \le \frac{b_1}{a_1}<\frac{b}{a} $}
symmetric to case 2, we get that the only ratio that can depend on $(a_1,b_1)$ and $(a_2,b_2)$ is $(a_1 * b_2 + b_1 * a_2 + b_2 * b_1, b_2 * b_1) = (a_1,b_1) *_3 (a_2,b_2)$
\end{proof}
All ratios are achievable by these three relations, along with the fact that for no cuts, $(0,1)$ and $(1,0)$ are achievable, as they correspond to giving it all to one participant. This also gives a method to, for any $n$, construct all ratios $(a,b)$ that have $f(a,b) \le n$. 

Notice that even if the two smaller ratios are in lowest terms, these products aren't necesarily in lowest terms. However, in showing the next proposition, we use $*_2$ and show that it does produce a ratio in lowest terms in a certain construction. 
\subsection{Bounds on $f(a,b)$}
\begin{construction}
There exists an infinite sequence of ratios $(a_n,b_n)$ in lowest terms with, $\forall n,  2^{2^{n-1}} = a_n +b_n$ and $f(a_n,b_n)\le n$\\*
\end{construction}
\begin{proof}
{\bf Base case}: take $n = 1$, then we can achieve $(1,1)$ in a single cut via cut-and-choose. This meets the assumption for $n=1$ that  $a+b = 1+1 = 2^{2^{1-1}} = 2^{2^{n-1}}$\\
{\bf Inductive step}: assume $(a_{n-1},b_{n-1})$ satisfy $2^{2^{n-2}} = a_{n-1} +b_{n-1}$ and are already in lowest terms.
We claim $(a_n,b_n) =  (a_{n-1},b_{n-1}) *_2  (b_{n-1},a_{n-1})$ satisfies the criteria.\\*
We have, $(a_{n-1},b_{n-1}) *_2  (b_{n-1},a_{n-1}) = (a_{n-1}  b_{n-1}, a_{n-1} b_{n-1} + a_{n-1}^2+ b_{n-1}^2)$ First, we show that this is in lowest terms. Assume not, that $1 \ne d =\gcd(a_{n-1} b_{n-1} , a_{n-1} b_{n-1} + a_{n-1}^2+ b_{n-1}^2)$ Then, take a prime $p|d$. Because $d|a_{n-1} b_{n-1}$ and $1 = \gcd(a_{n-1} , b_{n-1})$ either $p|a_{n-1}$ or $p|b_{n-1}$ but not both, wlog, consider, $p|a_{n-1}$. Then, $p |  a_{n-1} b_{n-1} + a_{n-1}^2+ b_{n-1}^2$. So, since $p |  a_{n-1} b_{n-1} + a_{n-1}^2$ we must have that $ p |  b_{n-1}^2$.
but, we have that $p \nmid b_{n-1}$, a contradiction, so, $d=1$, so, $(a_n,b_n)$ is in lowest terms. Now, we just note that $a_n +b_n = a_{n-1} b_{n-1} + a_{n-1} b_{n-1} + a_{n-1}^2+ b_{n-1}^2 = (a_{n-1} + b_{n-1})^2 = (2^{2^{n-2}})^2 = 2^{2^{n-1}}$. Lastly, since $f(a_{n-1},b_{n-1}) \le n-1$ and we are reducing $(a_n,b_n)$ to $(a_{n-1},b_{n-1})$ in a single cut, $f(a_{n},b_{n}) \le n$, completing the induction.
\end{proof}

\begin{corr}
 $\forall M, \exists a,b >M$ such that $\gcd(a,b)=1$ and $f(a,b) \le 1 + \lg(\lg(a+b))$
\end{corr}

\begin{construction}
For all ratios that can be acheived in n or less steps, (a,b) we have that $2^{2^{n}-1}\ge a+b$\\*
\end{construction}
\begin{proof}
{\bf Base case}: Take $n=1$, then the only acheivable ratios are $(1,1)$,$(0,1)$, and $(1,0)$ which satisfy the inequality.\\*
{\bf Inductive step}: Assume all ratios that can be done in $n$ steps satisfy the inequality, then, we know from Lemma ~\ref{buildupper} that all ratios that can be acheived in $n+1$ steps must be obtained from one of the three ways of combining ratios that can be acheived in $\le n$ steps. Let $(a,b)$ be a ratio with $f(a,b) = n+1$ steps. Then, we know from Lemma ~\ref{buildupper} that this depends on two ratios $(a_1,b_1)$ and $(a_2,b_2)$ with $f(a_1,b_1),f(a_2,b_2) \le n$ in one of three ways:\\
\caseTitle{Case 1:}
$(a,b) = (a_1,b_1) *_1 (a_2,b_2) =((a_1 +b_1) * a_2, (a_2 + b_2) * b_1) $ which has sum $(a_1 +b_1) *a_2 + (a_2 + b_2) * b_1 \le  2 *(a_1 +b_1) * (a_2 + b_2) \le 2*2^{2^{n}-1}*2^{2^{n}-1} = 2^{2^{n+1}-1}$\\*
\caseTitle{Case 2:}
$(a,b) = (a_1,b_1) *_2 (a_2,b_2) =(a_1 * a_2,  a_1 * a_2 + b_2 * a_1 + a_2 * b_1)$
which has sum $a_1 * a_2 + a_1 * a_2 + b_2 * a_1 + a_2 * b_1 \le a_1 * a_2 +(a_1 +b_1) * (a_2 + b_2) \le 2 * (a_1 +b_1) * (a_2 + b_2) \le 2*2^{2^{n}-1}*2^{2^{n}-1} = 2^{2^{n+1}-1}$\\*
\caseTitle{Case 3:}
$(a,b) = (a_1,b_1) *_3 (a_2,b_2) =(a_1 * b_2 + b_1 * a_2 + b_2 * b_1, b_2 * b_1)$ which has sum:
$a_1 * b_2 + b_1 * a_2 + b_2 * b_1 + b_2 * b_1 \le (a_1 +b_1) * (a_2 + b_2) + b_1*b_2 \le 2 *(a_1 +b_1) * (a_2 + b_2) \le 2*2^{2^{n}-1}*2^{2^{n}-1} = 2^{2^{n+1}-1}$
\end{proof}
\begin{corr}
For any ratio $(a,b)$ in lowest terms,  $f(a,b)  \ge \lg(1+\lg(a+b))$
\end{corr}

\bibliographystyle{plain}	
\bibliography{unequalbib}	
\end{document}